\documentclass[11pt]{article}

\usepackage[margin=1in]{geometry}
\usepackage{amsmath,amssymb,amsthm,mathtools}
\usepackage{booktabs,array,tabularx,longtable}
\usepackage{graphicx}
\graphicspath{{figs/}}
\usepackage{microtype}
\usepackage{xcolor}
\usepackage{enumitem}
\usepackage{caption}
\usepackage{float}
\usepackage{placeins}
\DeclareMathOperator*{\argmin}{argmin}
\usepackage{natbib}
\usepackage[colorlinks=true,linkcolor=blue!55!black,citecolor=blue!55!black,urlcolor=blue!55!black]{hyperref}

\definecolor{clinicalblue}{RGB}{16,47,85}
\definecolor{clinicalteal}{RGB}{35,139,123}
\definecolor{clinicalamber}{RGB}{185,111,18}
\definecolor{clinicalred}{RGB}{170,51,51}
\definecolor{clinicalgray}{RGB}{82,92,105}

\newtheorem{theorem}{Theorem}[section]
\newtheorem{proposition}[theorem]{Proposition}
\newtheorem{corollary}[theorem]{Corollary}
\newtheorem{lemma}[theorem]{Lemma}
\newtheorem{assumption}[theorem]{Assumption}
\theoremstyle{definition}

\theoremstyle{remark}

\newcommand{\E}{\mathbb E}
\renewcommand{\Pr}{\mathbb P}
\newcommand{\cF}{\mathcal F}
\newcommand{\1}{\mathbf 1}
\newcommand{\R}{\mathbb R}
\newcommand{\erf}{\operatorname{erf}}

\newcommand{\sprem}{\operatorname{sp}}

\title{\textbf{Mitigating the Winner’s Curse While Controlling Multiplicity: e-Process Methods for Anytime-Valid Inference in Dose-Ranging Trials}}
\author{\normalsize Victor K. de la Pe\~na$^{1,2,\dagger}$ \quad Fangyuan Lin$^{1,\dagger}$ \quad Demissie Alemayehu$^{1}$ \quad Victor H. de la Pe\~na$^1$\\[0.3em]
\small $^1$Department of Statistics, Columbia University \quad $^2$Infremacy LLC \\[-0.1em]
\small $^\dagger$These authors contributed equally.}
\date{June 2026}

\begin{document}
\maketitle

\begin{abstract}
Phase II dose-ranging trials commonly select the dose with the largest observed effect compared to control while repeatedly checking the data as it accumulates. This practice introduces two important problems: the “winner’s curse” (overly optimistic estimates for the selected dose) and inflated Type I error due to multiple comparisons across doses and interim analyses. In this paper, we propose a new statistical procedure that provides valid inference at any time during the trial. It tests whether the true effect of the best dose exceeds a clinically meaningful margin, while properly accounting for both the selection of the apparent winner and the repeated looks at the data. The method,  works by subtracting a “selection charge” that corrects for the expected optimism of always choosing the current best dose. This adjusted statistic is then monitored using an anytime-valid testing framework. The resulting decision rule is straightforward: take the observed best dose effect, subtract the selection charge and a monitoring margin, and declare “go” only if the remaining value still exceeds the target clinical margin. We provide ready-to-use implementations for normally distributed and binary outcomes. The procedure guarantees control of the Type I error rate in finite samples and supplies anytime-valid lower confidence bounds for the best dose effect. Its practical performance is demonstrated through worked examples and simulation studies.
\end{abstract}



\section{Introduction}\label{sec:intro}

\subsection{The decision problem}
Suppose $K$ active doses are randomized against one concurrent control. We use the following convention throughout:
\begin{center}
\begin{tabular}{@{}p{0.19\textwidth}p{0.72\textwidth}@{}}
\toprule
Symbol & Meaning\\
\midrule
$t$ & index for a sequential information update or analysis look\\
$n_t$ & cumulative number of evaluable subjects per arm at look $t$\\
$X_{t,k}$ & newly observed outcome from dose $k$ at update $t$; $X_{t,0}$ is the concurrent control outcome\\
$\varepsilon_{t,k}$ & centered noise, $X_{t,k}-\E(X_{t,k})$\\
\bottomrule
\end{tabular}
\end{center}

The balanced unit-block theory observes one outcome per arm and control at each update, so $n_t=t$ in that case. $X_{t,k}$ and $X_{t',k'}$ are assumed to be independent whenever $(t,k)\ne (t',k')$. Let
\[
\mu_k:=\E(X_{t,k}),\qquad \mu_0:=\E(X_{t,0}),\qquad
\delta_k:=\mu_k-\mu_0.
\]
The best program-level efficacy effect is
\begin{equation}\label{eq:delta-star}
\delta_*:=\max_{1\le k\le K}\delta_k.
\end{equation}
A protocol specifies a clinically meaningful margin $\delta_0$ and asks
\begin{equation}\label{eq:global-null}
H_0(\delta_0):\ \delta_*\le\delta_0
\qquad\text{versus}\qquad
H_1(\delta_0):\ \delta_*>\delta_0.
\end{equation}
This is a global proof-of-concept question: does the tested dose range contain at least one dose whose mean effect clears the clinical requirement? 

At look $t$, define cumulative outcomes
\[
U_{t,k}:=\sum_{s=1}^t X_{s,k},\qquad C_t:=\sum_{s=1}^t X_{s,0},
\]
and pairwise sample-mean differences
\[
\widehat\delta_{t,k}:=\frac{U_{t,k}-C_t}{n_t},\qquad
\widehat\delta_{\max,t}:=\max_k\widehat\delta_{t,k}
=\frac{\max_kU_{t,k}-C_t}{n_t}.
\]
The maximum is operationally natural to look at because a global go decision is often driven by whichever pairwise comparison currently looks strongest. It, however, creates two logically separate problems.
\begin{enumerate}[leftmargin=*,itemsep=2pt]
\item \textbf{Winner's curse/selection optimism.} Even when every $\widehat\delta_{t,k}$ is unbiased before selection, the largest one is optimistic because it was selected for being large.
\item \textbf{Multiplicity over doses and time.} Looking across $K$ doses and repeatedly over time creates multiple opportunities for a false global go decision.
\end{enumerate}
A multiplicity-adjusted test can solve the second problem while leaving the reported largest effect optimistic. On the other hand, subtracting a bias correction does not license repeated peeking. The paper therefore keeps bias correction and Type-I control as separate guarantees, even though the same selection-premium quantity drives both. Section~\ref{sec:bias-envelope} makes this bias-correction statement precise:
at the equal-margin reference, the accumulated selection premium is an exact
correction for the expected optimism of the selected maximum; under general
heterogeneous dose effects, the same charge gives an upper envelope for the
expected optimism relative to the true best effect $\delta_*$.

\subsection{Statistical guarantee}

For the present paper,
\begin{equation}\label{eq:global-type1}
\sup_{\delta_*\le\delta_0}
\Pr\{\text{a global go occurs at some look }t\}\le\alpha
\end{equation}
is the relevant Type-I requirement. It says that, when no tested dose truly clears the clinical margin, the probability of ever declaring that the program contains an efficacious dose is at most $\alpha$. 

Strong familywise error rate (FWER) is a different and stronger armwise requirement: under every configuration of true and false individual hypotheses, the probability of rejecting at least one true named-dose null is controlled. Our procedure makes no individual dose rejection, but authorizes a program-level efficacy conclusion; named-dose localization requires a another layer.

\subsection{Current approaches in Phase II dose-ranging trials}
\label{sec:current-approaches}

Phase II dose-ranging trials are usually learning-phase studies designed to determine
whether the drug has sufficient activity and to characterize the dose--response
relationship for later development. A common design is a randomized, parallel-group,
multi-arm trial comparing several active doses with a concurrent placebo or standard
control \citep{ICH1994E4}. Allocation is
often balanced, although modestly unbalanced allocation may be used to improve
information about placebo, safety, or clinically promising doses. Sample sizes are
typically chosen prospectively to achieve target operating characteristics, such as
80\% power for a clinically meaningful contrast on a binary or continuous endpoint \citep{piantadosi2024clinical}. Modern designs may also include
prespecified interim analyses, futility or efficacy stopping, dose dropping, or adaptive
dose allocation; these adaptations are usually implemented through group-sequential,
multi-arm multi-stage, combination-test, or Bayesian design principles
\citep{jennison1999group,FDAAdaptive2019,berry2010bayesian, magirr2012generalized, wason2012optimal}.

The corresponding analysis strategies fall into two broad classes. The first class uses
pairwise comparisons of each active dose with the common control. Multiplicity across
doses is then handled by many-to-one procedures such as Dunnett's test, Bonferroni or
Holm adjustments, or hierarchical testing strategies \citep{dunnett1955multiple,holm1979simple,FDA2022}.
These methods are natural when the primary question is whether at least one tested
dose separates from control, and they remain common because they are transparent and
make relatively few assumptions about the shape of the dose-response curve.

The second class uses model-based dose-response methods. The leading example is
MCP-Mod, which combines multiple contrast tests over a prespecified family of
candidate dose-response shapes with subsequent model fitting or model averaging
\citep{bretz2005combining,pinheiro2014model}. In the MCP step, the analysis tests for a non-flat
dose-response signal across candidate shapes such as linear, Emax, exponential, or
logistic models; in the Mod step, statistically supported models are used to estimate
the dose-response curve and select a target dose. FDA and EMA qualification documents
describe MCP-Mod as an appropriate structured methodology for Phase II dose finding
under model uncertainty, because it uses information across the active-dose continuum
and placebo rather than treating the doses only as unrelated pairwise comparisons
\citep{FDA2015MCPMod,EMA2014MCPMod}.

Dose-selection decisions usually have two logically distinct components. First, the
trial seeks proof of activity or proof of concept: evidence that the dose range contains
at least one dose with a clinically meaningful effect relative to control. Second, if such
evidence is obtained, the program selects a dose for further development using a broader
criterion, such as the smallest dose with clinically relevant effect, a high fraction of
the maximum modeled effect, the exposure-response plateau, safety and tolerability,
pharmacokinetics, or overall benefit--risk \citep{ICH1994E4,FDA2024DoseOptimization}.
The present paper addresses the first component. We develop an anytime-valid global
test for whether the best true effect among the tested doses exceeds a prespecified
clinical margin. For mathematical and interpretive clarity, we focus on pairwise
dose--control contrasts and do not rely on a trend test or a parametric dose-response
model.

\subsection{Related work}

\paragraph{Fixed-look many-to-one comparisons.}
Dunnett's procedure uses the correlation induced by a shared control to provide simultaneous many-to-one inference and strong FWER control at a fixed analysis \citep{dunnett1955multiple}. It is usually more efficient than a Bonferroni analysis that ignores this correlation. This solves dose multiplicity at the planned analysis, but it does not by itself license arbitrary repeated looks or quantify the optimism of the raw observed winner.

\paragraph{Preplanned group-sequential and MAMS designs.}
Treatment-selection group-sequential designs and generalized Dunnett/MAMS procedures extend shared-control testing to a finite, prespecified stage structure, including efficacy and futility boundaries and interim treatment selection \citep{stallard2003sequential,magirr2012generalized,wason2012optimal}. These methods can provide strong FWER control across arms and stages and may be preferable as a primary design when stage times, selection rules, and endpoints can be fully specified in advance. The \texttt{MAMS} software described by \citep{jaki2019r} implements several such designs and can also use a conditional-error calculation for unexpected modifications.

\paragraph{Combination tests and adaptive Dunnett tests.}
Adaptive combination tests combine valid stagewise $p$-values from disjoint increments. Adaptive Dunnett procedures use the conditional error of the original many-to-one test, while closure can extend these local tests to strong FWER control for a family of named-dose hypotheses \citep{bauer1994evaluation,koenig2008adaptive,friede2008comparison}. These methods are designed around stagewise tests and individual hypotheses.

\paragraph{Dose-response modeling.}
MCP-Mod combines a prespecified family of dose-response contrasts with model fitting to establish a dose-response signal and estimate a target dose under model uncertainty \citep{bretz2005combining,pinheiro2014model}. It exploits dose ordering and shape; our maximum-based global test deliberately does not. The two can be used in sequence or side by side, but they answer different scientific questions.

\paragraph{Inference after choosing a winner.}
Winner-inference methods target the effect of the selected named option and correct confidence intervals or point estimates for the selection event \citep{andrews2024inference}. Our inferential target is instead $\delta_*$, the best true effect available in the tested dose range, uniformly over analysis time. The target does not identify the best dose or give an exactly unbiased effect for the empirical leader.

\paragraph{Anytime-valid inference and betting.}
Confidence sequences and e-processes solve optional stopping in broad settings \citep{howard2021time,ramdas2023game}, and multi-armed betting tests provide anytime-valid global testing under other observation models \citep{sandoval2026multi}. What is new here is the use of the state-dependent selection premium as a predictable compensator/charge. It has the same units as the clinical effect and simultaneously serves as (i) an interpretable winner-optimism ledger and (ii) the predictable conditional-drift charge needed before safe exponentiation.

\section{The selection premium identity and the equal-margin reference}\label{sec:identity}

\subsection{The selection premium identity}
Let $X_t=(X_{t,1},\ldots,X_{t,K})$ denote the observed dose outcomes at update $t$ and write
$\varepsilon_t:=X_t-\E(X_t)$ for their centered noise. Let
\[
S_{t,k}:=\sum_{s=1}^t\varepsilon_{s,k},\qquad M_t:=\max_kS_{t,k}.
\]
The selection premium function (see \citep{de2026selection} for its properties) is a function of a state and an increment law. Let $\mathcal L$ denote the law of one fresh centered null update, and let $\eta=(\eta_1,\ldots,\eta_K)\sim\mathcal L$ be an independent copy of that next update. Define
\begin{equation}\label{eq:phi-foundation}
\sprem_{\mathcal L}(u):=\E_{\eta\sim\mathcal L}\!\left[\max_k(u_k+\eta_k)\right]-\max_k u_k,
\qquad u\in\R^K.
\end{equation}
The symbol $\eta$ is not an observed residual and not a second data set. It is the generic fresh noise used to compute the expectation of the next update. If the increments are i.i.d., the same function $\sprem_{\mathcal L}$ is used at every update. 

If one must commit now, the value is $\max_k u_k$. If one waits for one more update and may then choose the leader, the expected value is the first term in \eqref{eq:phi-foundation}. Thus $\sprem_{\mathcal L}(u)$ is the expected one-step value of retaining the option to switch leaders. For update $s$, write
\[
a_s:=\sprem_{\mathcal L}(S_{s-1}),\qquad A_t:=\sum_{s=1}^t a_s .
\]

The foundation paper \citep{de2026selection} proves, in the i.i.d. case,
\begin{equation}\label{eq:foundation-identity}
\E M_t=\sum_{s=1}^t\E\{\sprem_{\mathcal L}(S_{s-1})\}=\E\sum_{s=1}^t a_s,
\end{equation}
with a stopping-time analogue and a heterogeneous-mean extension \citep[Thms.~3.1, 3.5 and 4.3]{de2026selection}. It also proves translation invariance of $\sprem_{\mathcal L}$, maximality at the tie configuration, domination of heterogeneous-mean premiums by the equal-margin premium, and decay after a true leader emerges.

\subsection{Equal-margin configuration is a calibration reference}
For a candidate clinical effect $\delta$, the \emph{equal-margin reference} is the hypothetical model in which every dose has mean $\mu_0+\delta$, so every pairwise dose-control effect equals $\delta$. The actual trial can (and most likely) have arbitrary heterogeneous effects.

The reference is useful for two operational reasons.
\begin{enumerate}[leftmargin=*,itemsep=3pt]
\item \textbf{It isolates pure selection.} When dose means are equal, deterministic mean differences cannot explain the maximum. Its excess is entirely the value of being able to choose the noisiest leader.
\item \textbf{It gives the largest unknown-gap premium.} \citep{de2026selection} proves that true mean gaps only reduce the premium. When the gaps between the true dose effects are unknown, the equal-margin premium is therefore a conservative envelope for the winner's curse.
\end{enumerate}
The equal-margin calculation is therefore directly useful even when no investigator believes the true effects are exactly equal: it provides an exact benchmark for pure winner optimism, a worst-case charge when gaps between the true dose effects are unknown, and a state-dependent diagnostic of whether dose competition is still active.

Figure~\ref{fig:equalmeans} illustrates the behavior of the selection premium function $\sprem_{\mathcal L}$.

\begin{figure}[H]
\centering
\includegraphics[width=\textwidth]{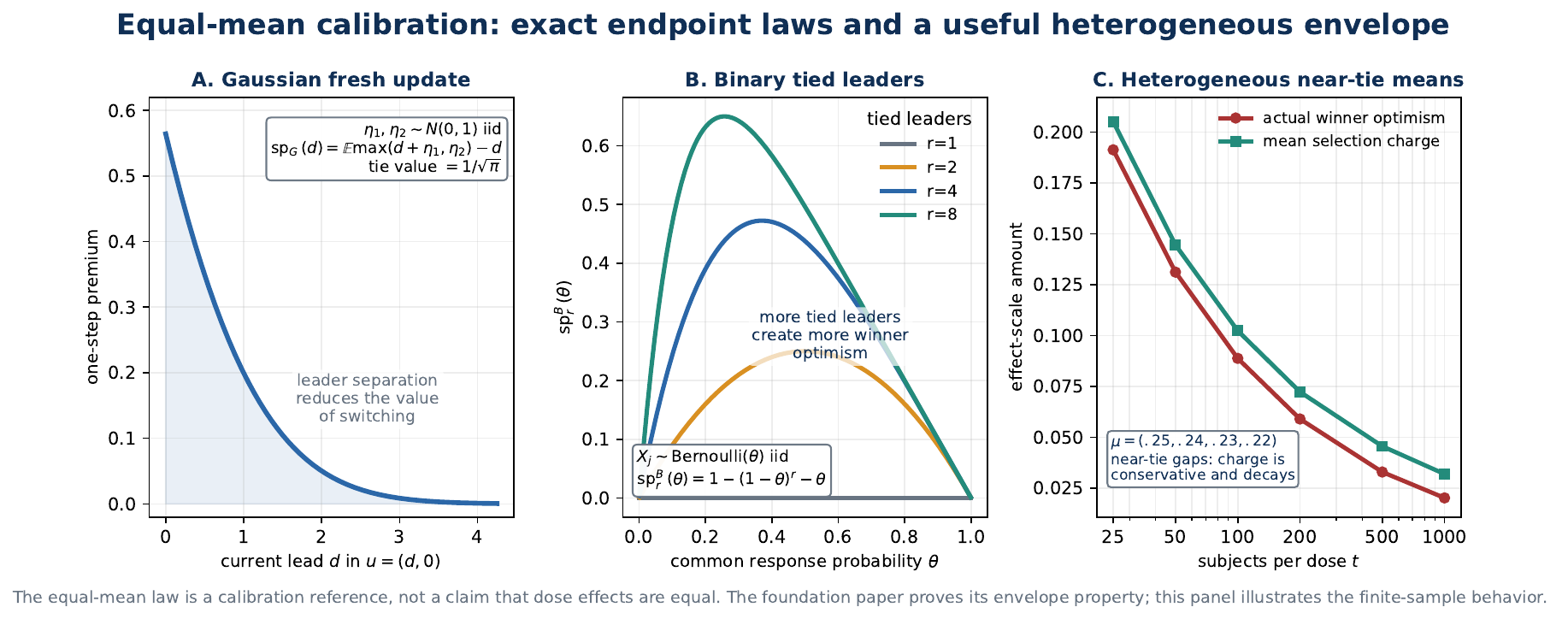}
\caption{Panel A uses $K=2$, state $u=(d,0)$, and i.i.d. $N(0,1)$ fresh dose noises. Panel B uses $r\in\{1,2,4,8\}$ tied leaders with i.i.d. Bernoulli$(\theta)$ outcomes; in the binary trial reference, $\theta=p_0+\delta$. Panel C uses a heterogeneous near-tie Gaussian example with $K=4$, independent $N(0,1)$ errors, and means $(0.25,0.24,0.23,0.22)$. The red curve is the exact expected fixed-look winner optimism $\E(\max_k\bar X_{t,k}-\mu_*)$. The green curve is the mean state-adaptive equal-margin selection charge $\E(A_t/t)$ from $3{,}500$ simulated paths; the shaded band is a Monte Carlo 95\% interval. Both curves decay with information.}
\label{fig:equalmeans}
\end{figure}

\section{From selection premium to e-process}\label{sec:generic}

\subsection{Cumulative score and corrected residual}
For any candidate effect $\delta$, define the cumulative global score at look $t$ by
\begin{equation}\label{eq:G}
G_t(\delta):=\max_{1\le k\le K}U_{t,k}-C_t-n_t\delta
=n_t\{\widehat\delta_{\max,t}-\delta\}.
\end{equation}
A large positive value means the current best pairwise estimate exceeds $\delta$. Due to winner's curse, it is not safe to monitor $G_t$ directly because the maximum operation creates positive selection drift.

Let $a_t(\delta)$ denote the selection charge computed from information available immediately before update $t$; equivalently, $a_t(\delta)=\sprem_{\mathcal L_t(\delta)}(\text{current state})$ for the appropriate null increment law and current score state. Define
\begin{equation}\label{eq:A-R-V}
A_t(\delta):=\sum_{s=1}^t a_s(\delta),\qquad
R_t(\delta):=G_t(\delta)-A_t(\delta),\qquad
V_t:=\sum_{s=1}^t v_s,
\end{equation}
where both $a_s(\delta)$ and $v_s$ are $\cF_{s-1}$-measurable, hence predictable. 

The next theorem is the bridge from the selection-premium identity to an anytime-valid decision rule. At the equal-effect reference, the running maximum has predictable conditional increment
\[
\E\{\Delta M_t\mid\cF_{t-1}\}=\sprem_{\mathcal L_t}(\text{current state})\ge0.
\]
Equivalently, the maximum is a submartingale and the evaluated selection premium is its predictable upward drift. Once this charge is subtracted, the residual has nonpositive conditional drift and can be turned into an e-process.

\begin{theorem}[Selection premium to e-process]\label{thm:generic}
Fix a candidate $\delta$. Suppose that under every distribution in $H_0(\delta)$,
\begin{align}
\E\{\Delta G_t(\delta)\mid\cF_{t-1}\}&\le a_t(\delta),\label{eq:drift-condition}\\
\E\!\left[\exp\left\{\lambda\bigl(\Delta G_t(\delta)-\E[\Delta G_t(\delta)\mid\cF_{t-1}]\bigr)\right\}\middle|\cF_{t-1}\right]
&\le \exp(\lambda^2v_t/2)\label{eq:mgf-condition}
\end{align}
for every $\lambda\ge0$, where $a_t(\delta)$ and $v_t$ are predictable. Then, for each fixed $\lambda\ge0$,
\begin{equation}\label{eq:L}
L_t(\lambda;\delta):=\exp\{\lambda R_t(\delta)-\lambda^2V_t/2\}
\end{equation}
is a nonnegative supermartingale with $L_0=1$.

If $\rho>0$ is fixed before monitoring and $L_t$ is mixed against the half-normal density
\[
\pi_\rho(\lambda)=\sqrt{2\rho/\pi}\exp(-\rho\lambda^2/2)\1\{\lambda\ge0\},
\]
then
\begin{equation}\label{eq:eprocess}
\mathcal E_t(\delta)
=\sqrt{\frac{\rho}{V_t+\rho}}
\left[1+\erf\left\{\frac{R_t(\delta)}{\sqrt{2(V_t+\rho)}}\right\}\right]
\exp\left\{\frac{R_t(\delta)^2}{2(V_t+\rho)}\right\}
\end{equation}
is an e-process. Consequently, for $\alpha > 0$
\begin{equation}\label{eq:ville}
\sup_{P\in H_0(\delta)}
\Pr_P\left\{\sup_{t\ge1}\mathcal E_t(\delta)\ge\frac1\alpha\right\}\le\alpha.
\end{equation}
\end{theorem}

\begin{proof}
Write
\[
b_t:=a_t(\delta)-\E\{\Delta G_t(\delta)\mid\cF_{t-1}\}\ge0
\]
and
$W_t:=\Delta G_t(\delta)-\E\{\Delta G_t(\delta)\mid\cF_{t-1}\}$.
Then $\Delta R_t=W_t-b_t$. For $\lambda\ge0$,
\[
\E(e^{\lambda\Delta R_t}\mid\cF_{t-1})
=e^{-\lambda b_t}\E(e^{\lambda W_t}\mid\cF_{t-1})
\le e^{-\lambda b_t}e^{\lambda^2v_t/2}
\le e^{\lambda^2v_t/2}.
\]
Multiplying by $L_{t-1}(\lambda;\delta)e^{-\lambda^2v_t/2}$ and conditioning proves
$\E\{L_t(\lambda;\delta)\mid\cF_{t-1}\}\le L_{t-1}(\lambda;\delta)$.

Nonnegative mixing preserves the supermartingale property by Tonelli's theorem; a different prespecified mixing law would give a different valid boundary. The closed form in \eqref{eq:eprocess} is the standard half-normal mixture calculation \citep{howard2020time}. Substituting \eqref{eq:L} and completing the square gives
\begin{align*}
\mathcal E_t(\delta)
&=\sqrt{\frac{2\rho}{\pi}}\int_0^\infty
\exp\left\{\lambda R_t-\frac{V_t+\rho}{2}\lambda^2\right\}d\lambda\\
&=\sqrt{\frac{2\rho}{\pi}}
\exp\left\{\frac{R_t^2}{2(V_t+\rho)}\right\}
\int_0^\infty
\exp\left[-\frac{V_t+\rho}{2}
\left\{\lambda-\frac{R_t}{V_t+\rho}\right\}^2\right]d\lambda.
\end{align*}
With $y=\sqrt{(V_t+\rho)/2}\{\lambda-R_t/(V_t+\rho)\}$, the lower integration limit is
$-R_t/\sqrt{2(V_t+\rho)}$, and the Gaussian tail integral is
$\sqrt{\pi/[2(V_t+\rho)]}[1+\erf\{R_t/\sqrt{2(V_t+\rho)}\}]$. This yields \eqref{eq:eprocess}. Ville's inequality gives \eqref{eq:ville}.
\end{proof}

\subsection{The GO rule}
For a variance budget $v\ge0$, define $q_\alpha(v;\rho)$ as the nonnegative solution $z$ of
\begin{equation}\label{eq:q}
\sqrt{\frac{\rho}{v+\rho}}
\left[1+\erf\left\{\frac{z}{\sqrt{2(v+\rho)}}\right\}\right]
\exp\left\{\frac{z^2}{2(v+\rho)}\right\}=\frac1\alpha.
\end{equation}
The solution exists and is unique. At $z=0$, the left side equals $\sqrt{\rho/(v+\rho)}\le1<1/\alpha$, while it diverges as $z\to\infty$. Moreover, for $z\ge0$ its log derivative is
\[
\frac{z}{v+\rho}
+\frac{\sqrt{2/\pi}\exp\{-z^2/[2(v+\rho)]\}}
{\sqrt{v+\rho}\,[1+\erf\{z/\sqrt{2(v+\rho)}\}]}>0,
\]
so the left side is strictly increasing.

At the protocol margin $\delta_0$,
\[
R_t(\delta_0)
=n_t\{\widehat\delta_{\max,t}-\delta_0\}-A_t(\delta_0).
\]
Because \eqref{eq:eprocess} is increasing in $R_t$ on the relevant nonnegative side,
$\mathcal E_t(\delta_0)\ge1/\alpha$ is equivalent to
$R_t(\delta_0)\ge q_\alpha(V_t;\rho)$. Dividing by $n_t$ yields the operational rule
\begin{equation}\label{eq:generic-decision}
\boxed{\quad
\widehat\delta_{\max,t}
\ge \delta_0+ \frac{A_t(\delta_0)}{n_t}+\frac{q_\alpha(V_t;\rho)}{n_t}.
\quad}
\end{equation}
Every term has a separate role: $\widehat\delta_{\max,t}$ is the raw best effect,
$A_t(\delta_0)/n_t$ is the selection-premium charge, $\delta_0$ is the clinical requirement, and
$q_\alpha(V_t;\rho)/n_t$ is the extra margin needed for safe repeated monitoring.

\subsection{Choosing the half-normal tuning constant \texorpdfstring{$\rho$}{rho} in practice}\label{sec:rho}
The parameter $\rho$ is the precision of the half-normal mixing distribution over the betting parameter $\lambda$. The choice of $\rho$ does not change Type I validity but can affect the exact Type I error.

For $v>0$, write $c=\rho/v$ and $z=\sqrt v\,x$ in \eqref{eq:q}. Then
\begin{equation}\label{eq:q-scale}
q_\alpha(v;\rho)=\sqrt v\,Q_\alpha(\rho/v),
\end{equation}
where $Q_\alpha(c):=q_\alpha(1;c)$. Hence only the ratio $\rho/v$ matters. If the trial team has a plausible variance budget $V_*$ at which a go decision would be most valuable, a natural targeted choice is
\begin{equation}\label{eq:rho-target}
\rho=c_\alpha^*V_*,\qquad
c_\alpha^*:=\argmin_{c>0}Q_\alpha(c).
\end{equation}
In the binary simulations below, $V_t=t/2$. The reported value $\rho=13.41$ targets $V_*\approx 89.7$, or about $179$ subjects per arm, because for $\alpha=0.05$ the minimizer in \eqref{eq:rho-target} is approximately $c_{0.05}^*=0.1494$. This tuning choice affects efficiency but not Type I validity.

\section{Gaussian outcomes}\label{sec:gaussian}

\subsection{Model and computable premium}
Gaussian outcomes cover continuous Phase II endpoints such as change from baseline in blood pressure, symptom score, or a biomarker. We use balanced unit blocks for clarity, so $n_t=t$ in the displayed formulas.

\begin{assumption}[Gaussian shared-control blocks]\label{ass:gaussian}
Across blocks, $(X_{t,1},\ldots,X_{t,K},X_{t,0})$ are independent. Within a block,
\[
X_{t,1:K}\sim N_K(\mu_{1:K},\Sigma_D),\qquad
X_{t,0}\sim N(\mu_0,\sigma_0^2),
\]
and the dose vector is independent of control. Before monitoring begins, the analysis specifies deterministic bounds
$\bar\Sigma_D\succeq\Sigma_D$ and $\bar\sigma_0^2\ge\sigma_0^2$.
\end{assumption}
$\bar\Sigma_D$ and $\bar\sigma_0^2$ can be based on external information or an independent variance-estimation stage.

For a centered Gaussian dose-noise vector $\varepsilon^G\sim N_K(0,\bar\Sigma_D)$, define the Gaussian selection premium
\begin{equation}\label{eq:gaussian-premium}
\sprem_G(u;\bar\Sigma_D):=
\E\!\left[\max_k(u_k+\varepsilon_k^G)\right]-\max_k u_k.
\end{equation}
Translation invariance makes it computable from the observed cumulative dose vector $U_{t-1}=(U_{t-1,1},\ldots,U_{t-1,K})$ without knowing $\mu_0$ or $\delta$. Set
\begin{equation}\label{eq:gaussian-charge}
A_t^G:=\sum_{s=1}^t\sprem_G(U_{s-1};\bar\Sigma_D),
\qquad
V_t^G:=t\bar v_G,
\qquad
\bar v_G:=\max_k(\bar\Sigma_D)_{kk}+\bar\sigma_0^2.
\end{equation}
The shared control is absent from $\sprem_G$ because a common subtraction cannot change which dose is largest.

For two doses, if the current score gap is $d=|u_1-u_2|$ and
$\omega^2=\operatorname{Var}(\varepsilon_1^G-\varepsilon_2^G)$, then
\begin{equation}\label{eq:gaussian-two}
\sprem_G(u;\bar\Sigma_D)=\omega\varphi(d/\omega)-d\{1-\Phi(d/\omega)\},
\end{equation}
where $\varphi$ and $\Phi$ are the standard normal density and distribution function. For independent doses with standard deviations $\sigma_k$, the premium for general $K$ can be computed via the one-dimensional density of the maximum,
\begin{equation}\label{eq:gaussian-integral}
\E\max_k(u_k+Y_k)
=\int_{-\infty}^{\infty}x
\sum_{k=1}^K\frac1{\sigma_k}\varphi\!\left(\frac{x-u_k}{\sigma_k}\right)
\prod_{j\ne k}\Phi\!\left(\frac{x-u_j}{\sigma_j}\right)dx.
\end{equation}
For correlated $\bar\Sigma_D$, evaluate \eqref{eq:gaussian-premium} directly by numerical integration or Monte Carlo.

\subsection{Gaussian decision rule}
The following proposition is the bias-accounting part of the Gaussian application. It tells a trial team that at the equal-margin reference, $A_t^G/n_t$ is the exact expected upward bias of the reported maximum and acts as an upper bound for the expected selection bias for arbitrary true effect configurations.

\begin{proposition}[Gaussian equal-margin bias identity]\label{prop:gaussian-bias}
Suppose Assumption~\ref{ass:gaussian} holds with $\bar\Sigma_D=\Sigma_D$ and $\bar\sigma_0^2=\sigma_0^2$, and suppose
$\mu_1-\mu_0=\cdots=\mu_K-\mu_0=\delta$. Then for every fixed look $t$,
\begin{equation}\label{eq:gaussian-exact-bias}
\E\{\widehat\delta_{\max,t}-\delta\}=\frac1{n_t}\E(A_t^G).
\end{equation}
For every bounded stopping time $\tau$,
\begin{equation}\label{eq:gaussian-stop}
\E\{\max_kU_{\tau,k}-C_\tau-n_\tau\delta\}=\E(A_\tau^G).
\end{equation}
\end{proposition}
\begin{proof}
Under the equal-margin reference, $U_{t,k}-n_t(\mu_0+\delta)$ are centered Gaussian dose walks. Translation invariance identifies their one-step premium with \eqref{eq:gaussian-premium} evaluated at $U_{t-1}$. Apply the fixed- and stopped-horizon identities from \citet{de2026selection}, then subtract the centered control sum.
\end{proof}

The next theorem gives Type I error control. The equal-margin premium becomes a safe predictable drift charge over the entire composite null, including heterogeneous configurations in which some doses are below the margin.

\begin{theorem}[Gaussian anytime-valid global test]\label{thm:gaussian}
Under Assumption~\ref{ass:gaussian}, for every candidate $\delta$ and every parameter vector satisfying $\delta_k\le\delta$ for all $k$,
\begin{equation}\label{eq:gaussian-drift}
\E\{\Delta G_t(\delta)\mid\cF_{t-1}\}
\le \sprem_G(U_{t-1};\bar\Sigma_D).
\end{equation}
Moreover, the centered increment of $G_t(\delta)$ is conditionally sub-Gaussian with variance proxy $\bar v_G$. Therefore Theorem~\ref{thm:generic}, with $A_t=A_t^G$ and $V_t=V_t^G$, yields an e-process $\mathcal E_t^G(\delta)$ and
\[
\sup_{P\in H_0(\delta)}
\Pr_P\left\{\sup_t\mathcal E_t^G(\delta)\ge1/\alpha\right\}\le\alpha.
\]
Equivalently, declaring GO at the first look satisfying
\begin{equation}\label{eq:gaussian-decision}
\boxed{\quad
\widehat\delta_{\max,t}
\ge\delta_0+\frac{A_t^G}{n_t}+\frac{q_\alpha(t\bar v_G;\rho)}{n_t}
\quad}
\end{equation}
yields a level-$\alpha$ global test of $H_0(\delta_0)$ at arbitrary data-dependent looks.
\end{theorem}
\begin{proof}
Given $\cF_{t-1}$,
\[
\Delta G_t(\delta)=\max_k(U_{t-1,k}+X_{t,k})-\max_kU_{t-1,k}-X_{t,0}-\delta .
\]
Let $\theta=\mu_0+\delta$ and $d_k=\mu_k-\theta=\delta_k-\delta\le0$. With $\varepsilon_t=X_{t,1:K}-\mu_{1:K}$,
\begin{align*}
\E\{\Delta G_t(\delta)\mid\cF_{t-1}\}
&=\E\max_k(U_{t-1,k}+d_k+\varepsilon_{t,k})-\max_kU_{t-1,k}\\
&\le \E\max_k(U_{t-1,k}+\varepsilon_{t,k})-\max_kU_{t-1,k}.
\end{align*}
The inequality uses coordinatewise monotonicity of the maximum. If $\bar\Sigma_D\succeq\Sigma_D$, write $\bar\varepsilon\stackrel{d}=\varepsilon_t+\zeta$ with $\zeta\sim N_K(0,\bar\Sigma_D-\Sigma_D)$ independent. Conditional Jensen's inequality for the convex maximum gives
\[
\E_\zeta\{\max_k(U_{t-1,k}+\varepsilon_{t,k}+\zeta_k)\mid\varepsilon_t\}
\ge \max_k(U_{t-1,k}+\varepsilon_{t,k}),
\]
so the bound is at most $\sprem_G(U_{t-1};\bar\Sigma_D)$.

For the second condition, the centered dose-maximum increment is the supremum of a finite
centered Gaussian process. By Borell's inequality for Gaussian
process suprema, it is sub-Gaussian with variance proxy
\(\sup_k \operatorname{Var}(\varepsilon_{t,k})\le \max_k(\bar\Sigma_D)_{kk}\).
\end{proof}

\subsection{Finite selection cost under a unique best Gaussian dose}
\begin{proposition}[Gaussian charge is summable after separation]\label{prop:gaussian-summable}
Under Assumption~\ref{ass:gaussian}, suppose a unique best dose $k^*$ has
$\mu_{k^*}>\max_{j\ne k^*}\mu_j$. Then
\begin{equation}\label{eq:gaussian-finite}
A_\infty^G:=\sum_{t=1}^{\infty}\sprem_G(U_{t-1};\bar\Sigma_D)<\infty
\qquad\text{almost surely}.
\end{equation}
Consequently, $A_t^G/n_t\to0$ almost surely.
\end{proposition}
\begin{proof}
By the strong law, $U_{t,k^*}-U_{t,j}$ grows linearly and is eventually positive for every $j\ne k^*$. Once $k^*$ is the current leader, let $d_{t,j}=U_{t,k^*}-U_{t,j}$. For $\varepsilon^G\sim N_K(0,\bar\Sigma_D)$,
\[
\sprem_G(U_t;\bar\Sigma_D)
\le\sum_{j\ne k^*}\E[(\varepsilon_j^G-\varepsilon_{k^*}^G-d_{t,j})_+].
\]
If $\omega_j^2=\operatorname{Var}(\varepsilon_j^G-\varepsilon_{k^*}^G)$, then
$\E[(N(0,\omega_j^2)-d)_+]\le \omega_j\varphi(d/\omega_j)$ for $d>0$. Linear growth of $d_{t,j}$ therefore makes the displayed upper bound summable.
\end{proof}

\section{Binary outcomes}\label{sec:binary}

\subsection{Model and closed-form tie premium}
Binary outcomes cover response, remission, and other success-failure endpoints. Here the unit information update contains one new binary outcome from each dose and one from control, so the per-arm sample size is $n_t=t$.

\begin{assumption}[Binary shared-control blocks]\label{ass:binary}
At update $t$, $X_{t,k}\in\{0,1\}$ is observed from dose $k$ and $X_{t,0}\in\{0,1\}$ from control. Updates are independent. Within an update the $K+1$ outcomes are independent, with
$X_{t,k}\sim\operatorname{Bernoulli}(p_k)$ and $X_{t,0}\sim\operatorname{Bernoulli}(p_0)$.
\end{assumption}
As in Section~\ref{sec:intro}, define the cumulative dose and control counts by
\[
U_{t,k}:=\sum_{s\le t}X_{s,k},\qquad C_t:=\sum_{s\le t}X_{s,0}.
\]
Immediately before update $t$, let
\begin{equation}\label{eq:leaders}
\mathcal K_t:=\{k:U_{t-1,k}=\max_jU_{t-1,j}\},\qquad r_t:=|\mathcal K_t|,
\end{equation}
where $r_t$ is the number of tied leaders before update $t$.

Only a current leader can increase the dose maximum in one binary update. Define
\[
B_t:=\1\!\left\{\sum_{k\in \mathcal K_t}X_{t,k}\ge1\right\},
\]
the indicator that at least one current leader responds.

For candidate margin $\delta\in[-1,1]$, the equal-margin reference uses dose response probability $\theta$ and control probability $\theta-\delta$. Feasibility requires
\begin{equation}\label{eq:binary-feasible}
\theta\in\Theta(\delta):=[\max(0,\delta),\,\min(1,1+\delta)].
\end{equation}
The exact one-step selection premium is
\begin{equation}\label{eq:binary-premium}
\sprem_r^B(\theta):=\Pr(B_t=1\mid\cF_{t-1})-\theta
=1-(1-\theta)^r-\theta.
\end{equation}
This is the same object as the generic $\sprem_{\mathcal L}$ in \eqref{eq:phi-foundation}, specialized to a binary state with $r$ tied leaders. When $r=1$, $\sprem_1^B(\theta)=0$: a unique leader creates no immediate option value. When several leaders are tied, the maximum rises if any of them responds, and the option to switch leaders has positive value.

Because $\theta$ is unknown in routine use, define the nuisance-robust premium
\begin{equation}\label{eq:binary-robust}
\overline{\sprem}_r^B(\delta):=\sup_{\theta\in\Theta(\delta)}\sprem_r^B(\theta),\qquad
\bar A_t^B(\delta):=\sum_{s=1}^t\overline{\sprem}_{r_s}^B(\delta),\qquad
V_t^B:=t/2.
\end{equation}

The binary case is attractive operationally because the one-step charge is closed form. The only state variable needed for the selection premium is $r_t$, the number of current response-count leaders before the next update. The lemma below removes the remaining nuisance parameter by maximizing over the feasible boundary response probability.

\begin{lemma}[Closed-form binary charge]\label{lem:binary-closed}
For $r=1$, $\overline{\sprem}_1^B(\delta)=0$. For $r\ge2$, $\sprem_r^B$ is strictly concave on $[0,1)$ and has unconstrained maximizer
\[
\theta_r^*=1-r^{-1/(r-1)}.
\]
Let $\ell(\delta)=\max(0,\delta)$ and $u(\delta)=\min(1,1+\delta)$. Then
\begin{equation}\label{eq:binary-closed}
\overline{\sprem}_r^B(\delta)=\sprem_r^B\!\left(\min\{u(\delta),\max\{\ell(\delta),\theta_r^*\}\}\right).
\end{equation}
Moreover $h_r(\delta):=\delta+\overline{\sprem}_r^B(\delta)$ is nondecreasing on $[-1,1]$.
\end{lemma}
\begin{proof}
Differentiate:
$\{\sprem_r^B\}'(\theta)=r(1-\theta)^{r-1}-1$ and
$\{\sprem_r^B\}''(\theta)=-r(r-1)(1-\theta)^{r-2}\le0$, with strict inequality on $[0,1)$ for $r\ge2$. Therefore the constrained maximizer is the projection of $\theta_r^*$ onto $\Theta(\delta)$, giving \eqref{eq:binary-closed}. For monotonicity, first note that $h_1(\delta)=\delta$. For $r\ge2$, consider the feasible interval separately. If $\delta\ge0$, then $\Theta(\delta)=[\delta,1]$. When $\delta\le\theta_r^*$, the optimizer is interior and $h_r'(\delta)=1$; when $\delta>\theta_r^*$, the optimizer is $\theta=\delta$ and
\[
h_r(\delta)=1-(1-\delta)^r,
\]
whose derivative is $r(1-\delta)^{r-1}\ge0$. If $\delta<0$, then $\Theta(\delta)=[0,1+\delta]$. When $1+\delta\ge\theta_r^*$, again $h_r'(\delta)=1$; otherwise the optimizer is $\theta=1+\delta$ and
\[
h_r(\delta)=-(-\delta)^r,
\]
whose derivative is $r(-\delta)^{r-1}\ge0$. Continuity at the case boundaries completes the proof.
\end{proof}

\subsection{Exact reference bias and anytime-valid test}
The binary bias identity has the same interpretation as the Gaussian one, but it is even more transparent: ties among current response-count leaders create the winner's curse, and a unique leader creates no one-step binary selection charge.

\begin{proposition}[Binary equal-margin bias identity]\label{prop:binary-bias}
Under Assumption~\ref{ass:binary}, suppose $p_1=\cdots=p_K=\theta$ and $p_0=\theta-\delta$. Define
$A_t^B(\theta):=\sum_{s=1}^t\sprem_{r_s}^B(\theta)$. Then
\begin{equation}\label{eq:binary-exact-bias}
\E\{\widehat\delta_{\max,t}-\delta\}=\frac1{n_t}\E\{A_t^B(\theta)\}
\end{equation}
for every fixed look $t$, and
\begin{equation}\label{eq:binary-stop}
\E\{\max_kU_{\tau,k}-C_\tau-n_\tau\delta\}=\E\{A_\tau^B(\theta)\}
\end{equation}
for bounded stopping times $\tau$; in the unit-block model $n_\tau=\tau$.
\end{proposition}
\begin{proof}
At the equal-margin reference, the centered cumulative dose counts form a random walk whose one-step maximum increment has conditional mean $\sprem_{r_t}^B(\theta)$. The fixed-time identity follows by summing these predictable increments and subtracting the centered control contribution. The stopped identity follows from the stopped version of the same selection-premium identity for bounded stopping times \citep{de2026selection}.
\end{proof}

\begin{theorem}[Binary anytime-valid global test]\label{thm:binary}
Under Assumption~\ref{ass:binary}, for every candidate $\delta\in[-1,1]$ and every configuration with $p_k-p_0\le\delta$ for all $k$,
\begin{equation}\label{eq:binary-drift}
\E\{\Delta G_t(\delta)\mid\cF_{t-1}\}
\le\overline{\sprem}_{r_t}^B(\delta).
\end{equation}
The centered increment is conditionally sub-Gaussian with variance proxy $1/2$. Hence Theorem~\ref{thm:generic}, with $A_t=\bar A_t^B(\delta)$ and $V_t=t/2$, gives an e-process $\mathcal E_t^B(\delta)$ and an anytime-valid level-$\alpha$ global test. At the clinical margin,
\begin{equation}\label{eq:binary-decision}
\boxed{\quad
\widehat\delta_{\max,t}
\ge\delta_0+\frac{\bar A_t^B(\delta_0)}{n_t}+\frac{q_\alpha(t/2;\rho)}{n_t},\qquad n_t=t.
\quad}
\end{equation}
\end{theorem}
\begin{proof}
Given $\cF_{t-1}$, the score increment is
$\Delta G_t(\delta)=B_t-X_{t,0}-\delta$. If $p_0+\delta\le1$, the null configuration is feasible only when $p_0+\delta\ge0$; otherwise $p_k\le p_0+\delta<0$ for all $k$, contradicting $p_k\ge0$. Thus $\theta=p_0+\delta\in\Theta(\delta)$. Since $p_k\le p_0+\delta$ under the null,
\begin{align*}
\E\{\Delta G_t(\delta)\mid\cF_{t-1}\}
&=1-\prod_{k\in \mathcal K_t}(1-p_k)-p_0-\delta\\
&\le1-(1-p_0-\delta)^{r_t}-(p_0+\delta)\\
&=\sprem_{r_t}^B(\theta)\le\overline{\sprem}_{r_t}^B(\delta).
\end{align*}
If $p_0+\delta>1$, the conditional mean is at most $1-p_0-\delta<0\le\overline{\sprem}_{r_t}^B(\delta)$. The centered variables $B_t-\E(B_t\mid\cF_{t-1})$ and $X_{t,0}-p_0$ are independent and each is $1/4$-sub-Gaussian by Hoeffding's lemma. Their difference has proxy $1/2$. Apply Theorem~\ref{thm:generic}.
\end{proof}

\begin{proposition}[Finite binary selection charge under a unique best dose]\label{prop:binary-finite}
If one binary dose has strictly larger response probability than every other dose, then for every fixed candidate $\delta$,
\[
\bar A_\infty^B(\delta)<\infty\qquad\text{almost surely}.
\]
\end{proposition}
\begin{proof}
The strong law implies that the cumulative response count of the unique best dose eventually exceeds every other dose count and remains ahead. Thereafter $r_t=1$ for all sufficiently large $t$, and $\overline{\sprem}_1^B(\delta)=0$.
\end{proof}

\subsection{Exact equal-margin correction and heterogeneous-effect envelope}
\label{sec:bias-envelope}

Propositions~\ref{prop:gaussian-bias} and \ref{prop:binary-bias} give the
exact winner-optimism correction at the equal-effect reference. At that
reference, all active doses have the same dose--control effect $\delta$, so the
excess of the selected maximum is pure selection bias rather than a consequence
of genuine dose separation. 

The same quantity also gives a conservative envelope away from the equal-margin
case. For the actual trial, the dose effects may be heterogeneous. Define the
endpoint-specific robust charge $A_t(\delta)$ as in Sections~\ref{sec:gaussian}
and \ref{sec:binary}. At the unknown best effect $\delta_*$, every dose satisfies
$\delta_k\le\delta_*$. Thus the drift calculations used in the Gaussian and
binary e-process constructions imply the following upper bound.

\begin{proposition}[Heterogeneous best-effect bias envelope]\label{prop:hetero-envelope}
For either endpoint model and every fixed look $t$,
\begin{equation}\label{eq:hetero-envelope}
\E\left[n_t\{\widehat\delta_{\max,t}-\delta_*\}\right]
\le
\E\{A_t(\delta_*)\}.
\end{equation}
The same inequality holds at bounded stopping times on the cumulative score scale.
\end{proposition}

\begin{proof}
At $\delta=\delta_*$, the composite null inequalities $\delta_k\le\delta_*$
hold for all doses. Summing the drift bounds \eqref{eq:gaussian-drift} or
\eqref{eq:binary-drift} over $s\le t$ gives
\[
\E\{G_t(\delta_*)\}\le \E\{A_t(\delta_*)\}.
\]
Since $G_t(\delta_*)=n_t\{\widehat\delta_{\max,t}-\delta_*\}$, the fixed-time
inequality follows. For bounded stopping times, apply the same argument to the
stopped residual decomposition and use optional sampling.
\end{proof}

\section{Anytime lower confidence sequence for the best effect}\label{sec:cs}

For each candidate $\delta$, compute the corresponding e-process $\mathcal E_t(\delta)$. The rejected candidate values are those for which the data provide anytime-valid evidence that $\delta_*>\delta$.

Testing one candidate margin is useful for a GO/no-GO decision. Inverting all candidate margins gives a more interpretable object: an anytime lower floor for the best effect available in the tested dose range.

\begin{theorem}[Anytime lower confidence sequence]\label{thm:cs}
Assume either the Gaussian or binary model. Define
\begin{equation}\label{eq:lower-cs}
\underline\delta_t:=\sup\{\delta\in\mathcal D:\mathcal E_t(\delta)\ge1/\alpha\},
\end{equation}
where $\mathcal D=\R$ for Gaussian outcomes and $\mathcal D=[-1,1]$ for the displayed binary risk-difference analysis, with the supremum of the empty set taken as the lower endpoint of $\mathcal D$. Then
\begin{equation}\label{eq:cs-coverage}
\Pr\{\underline\delta_t\le\delta_*\text{ for every }t\ge1\}\ge1-\alpha.
\end{equation}
\end{theorem}
\begin{proof}
If $\underline\delta_t>\delta_*$ for some $t$, then the level-$\alpha$ test has rejected the true null $H_0(\delta_*)$ at some time. The probability of this event is at most $\alpha$ by Theorems~\ref{thm:gaussian} or \ref{thm:binary}.
\end{proof}

For Gaussian outcomes, $\mathcal E_t(\delta)$ is nonincreasing in $\delta$ because the charge is independent of $\delta$. For binary outcomes, $n_t\delta+\bar A_t^B(\delta)=\sum_{s=1}^t\{\delta+\overline{\sprem}_{r_s}^B(\delta)\}$ is nondecreasing by Lemma~\ref{lem:binary-closed}; hence $\mathcal E_t(\delta)$ is also nonincreasing in $\delta$. Thus \eqref{eq:lower-cs} can be obtained by one-dimensional bisection. The global GO rule at the clinical margin is simply
\[
\underline\delta_t\ge\delta_0.
\]

The finite-sample rule is conservative early because it pays for both selection and repeated monitoring. The asymptotic result below explains why this cost is transient when one dose is truly best: the selection ledger per subject vanishes, the monitoring boundary per subject vanishes, and the lower floor approaches the true best effect.

\begin{corollary}[Consistency and evidence growth]\label{cor:consistency}
Suppose a unique best dose exists and $\eta:=\delta_* -\delta_0>0$. Under either endpoint model,
\[
\underline\delta_t\to\delta_*\quad\text{almost surely},
\]
and the global test crosses in finite time almost surely. In the Gaussian model,
\[
\frac1{n_t}\log\mathcal E_t^G(\delta_0)\to\frac{\eta^2}{2\bar v_G}
\qquad\text{almost surely};
\]
in the binary model,
\[
\frac1{n_t}\log\mathcal E_t^B(\delta_0)\to\eta^2
\qquad\text{almost surely}.
\]
\end{corollary}
\begin{proof}
The strong law gives $\widehat\delta_{\max,t}\to\delta_*$. Propositions~\ref{prop:gaussian-summable} and \ref{prop:binary-finite} give $A_t/n_t\to0$. The half-normal mixture boundary satisfies $q_\alpha(V_t;\rho)=O\{\sqrt{V_t\log V_t}\}$, hence $q_\alpha(V_t;\rho)/n_t\to0$. The lower sequence therefore converges to $\delta_*$. In the binary case, the inversion uses the monotonicity of $\delta+\overline{\sprem}_r^B(\delta)$ from Lemma~\ref{lem:binary-closed}; the same vanishing-charge argument applies locally around $\delta_*$. Substituting $R_t(\delta_0)=n_t\eta+o(n_t)$ into \eqref{eq:eprocess} gives the stated almost-sure log-growth limits.
\end{proof}

\section{A binary example}\label{sec:worked-binary}

Suppose the control response probability is $0.30$ and the six dose response probabilities are
\[
(0.32,0.38,0.44,0.50,0.47,0.42),
\]
so the true dose-control effects are $(0.02,0.08,0.14,0.20,0.17,0.12)$. The clinical margin is $\delta_0=0.05$, the anytime level is $\alpha=0.05$, and $\rho=13.41$. As noted in Section~\ref{sec:rho}, this value targets an information budget near $179$ subjects per arm on the binary $V_t=t/2$ scale; it is not needed for validity. One outcome is observed from each dose and the control at every update, so in this balanced illustration $t=n_t$ is also the number of subjects per arm.

Before observing update $s$, the software counts the number $r_s$ of doses tied for the cumulative response-count lead and adds the robust binary premium $\overline{\sprem}^{B}_{r_s}(\delta_0)$ to the running ledger
\[
\bar A_t^B(\delta_0)=\sum_{s=1}^t \overline{\sprem}^{B}_{r_s}(\delta_0).
\]
After the outcomes at look $t$ arrive, it computes the raw best risk difference, the selection charge $\bar A_t^B/t$, the monitoring margin $q_{0.05}(t/2;13.41)/t$, and their difference. Table~\ref{tab:worked-path} gives the quantities a data-monitoring report needs.

\begin{table}[H]
\centering
\caption{One representative path. The first global GO occurs at $t=216$ subjects per arm. Percentages are risk-difference points.}
\label{tab:worked-path}
\small
\begin{tabularx}{\textwidth}{@{}c p{0.19\textwidth} *{4}{>{\centering\arraybackslash}X} c@{}}
\toprule
$t=n_t$ & Current leader count versus control & Raw best & Selection charge & Monitoring margin & Anytime lower floor & Decision\\
\midrule
50  & Dose 4: $29/50$ vs. $18/50$ & $22.0\%$ & $5.1\%$ & $29.5\%$ & $-12.6\%$ & Continue\\
100 & Dose 4: $54/100$ vs. $33/100$ & $21.0\%$ & $6.3\%$ & $19.8\%$ & $-5.1\%$ & Continue\\
150 & Doses 3 and 4: $77/150$ vs. $48/150$ & $19.3\%$ & $6.5\%$ & $16.0\%$ & $-3.2\%$ & Continue\\
200 & Dose 4: $109/200$ vs. $64/200$ & $22.5\%$ & $5.0\%$ & $13.8\%$ & $3.7\%$ & Continue\\
216 & Dose 4: $116/216$ vs. $66/216$ & $23.15\%$ & $4.64\%$ & $13.32\%$ & $5.19\%$ & \textbf{GO}\\
\bottomrule
\end{tabularx}
\end{table}

For this path, the reported lower floor equals raw best minus charge minus monitoring margin because the robust binary charge is locally constant over the relevant values of $\delta$; in general Theorem~\ref{thm:cs} defines the floor by direct inversion of $\mathcal E_t(\delta)$.

At the crossing look,
\begin{equation}\label{eq:worked-simple}
\underbrace{0.2315}_{\text{raw best risk difference}}
-\underbrace{0.0464}_{\text{selection charge}}
-\underbrace{0.1332}_{\text{monitoring margin}}
=\underbrace{0.0519}_{\text{anytime lower floor}}
>\underbrace{0.0500}_{\text{clinical margin}}.
\end{equation}
The trial therefore stops for a global proof-of-activity conclusion. Up to the crossing, $179$ pre-update states had one leader, $34$ had a two-way tie, $2$ had a four-way tie, and $1$ had a six-way tie. These states produce the cumulative charge $\bar A_{216}^B=10.027$. The practical message is simple: early dose competition is charged, while a persistent unique leader adds no further binary tie charge.

\section{Simulation}\label{sec:sim}

At update $t$, one new binary outcome is generated in each of the six dose groups and in the common control group:
\begin{equation}\label{eq:sim-dgm}
X_{t,0}\sim\operatorname{Bernoulli}(p_0),\qquad
X_{t,k}\sim\operatorname{Bernoulli}(p_k),\quad k=1,\ldots,6,
\end{equation}
independently across groups and updates, with $p_0=0.30$. The clinical margin is $\delta_0=0.05$, the anytime level is $\alpha=0.05$, and the prespecified mixture tuning parameter is $\rho=13.41$.

Table~\ref{tab:sim-scenarios} shows the three models used in the simulation.

\begin{table}[H]
\centering
\caption{Three models with different levels of effect signal. The control probability is $p_0=0.30$ in every row. The dose effects are $\delta_k=p_k-p_0$, $\delta_*=\max_k\delta_k$, and $\eta=\delta_* - 0.05$ is the amount by which the best true effect exceeds the clinical margin.}
\label{tab:sim-scenarios}
\small
\begin{tabularx}{\textwidth}{@{}p{0.17\textwidth}XXcc@{}}
\toprule
Scenario & Dose response probabilities $(p_1,\ldots,p_6)$ & True risk differences $(\delta_1,\ldots,\delta_6)$ & $\delta_*$ & $\eta$\\
\midrule
Near threshold & $(0.31,0.35,0.38,0.40,0.39,0.36)$ & $(0.01,0.05,0.08,0.10,0.09,0.06)$ & $0.10$ & $0.05$\\
Moderate & $(0.32,0.37,0.41,0.45,0.43,0.39)$ & $(0.02,0.07,0.11,0.15,0.13,0.09)$ & $0.15$ & $0.10$\\
Strong & $(0.32,0.38,0.44,0.50,0.47,0.42)$ & $(0.02,0.08,0.14,0.20,0.17,0.12)$ & $0.20$ & $0.15$\\
\bottomrule
\end{tabularx}
\end{table}

Immediately before update $t$, the rule counts the number $r_t$ of dose groups tied for the current cumulative response-count lead and adds the nuisance-robust binary premium $\overline{\sprem}^{B}_{r_t}(\delta_0)$ to the ledger. After the new outcomes are observed, it declares the first global GO when
\begin{equation}\label{eq:sim-go-clear}
\widehat\delta_{\max,t}
-\frac{\bar A_t^B(\delta_0)}{t}
-\frac{q_{0.05}(t/2;13.41)}{t}
\ge 0.05.
\end{equation}
The cap of $500$ per arm is a simulation horizon, not a proposed sample size for a Phase II trial.

\subsection{Finite-horizon power of the anytime rule}
For an anytime rule, ``power'' must specify a horizon. We report
\[
\operatorname{Power}_{T}(p):=\Pr_p\{\tau\le T\},
\]
where $\tau$ is the first global GO time and $T$ is the maximum subjects per arm. This is the probability that the trial reaches a valid global GO by the operational cap. Table~\ref{tab:power-grid} uses the same six-dose shared-control binary design, continuous monitoring, $\alpha=0.05$, $\rho=13.41$, and $T=500$.

\begin{table}[H]
\centering
\caption{Finite-horizon power of the proposed anytime selection-premium rule. Each alternative uses $20{,}000$ trials and the same models as Table~\ref{tab:sim-scenarios}.}
\label{tab:power-grid}
\small
\begin{tabularx}{\textwidth}{@{}p{0.19\textwidth}cccccX@{}}
\toprule
Scenario & $\delta_*$ & $\eta$ & GO by 100 & GO by 200 & GO by 500 & 50\% GO time\\
\midrule
Near threshold & $0.10$ & $0.05$ & $1.7\%$ & $4.3\%$ & $11.5\%$ & Not reached\\
Moderate & $0.15$ & $0.10$ & $6.4\%$ & $19.0\%$ & $57.1\%$ & $440$ per arm\\
Strong & $0.20$ & $0.15$ & $18.7\%$ & $50.6\%$ & $94.8\%$ & $198$ per arm\\
\bottomrule
\end{tabularx}
\end{table}

As the best effect moves farther above the margin, both the GO probability and the time to GO improve. This is the desired behavior: larger true gaps create persistent leaders, which reduce future selection charges and make the corrected evidence grow faster.

\subsection{Comparison to other procedures}
We used the same six-dose shared-control binary model as the main experiment. Under the all-boundary null,
\[
X_{t,0}\sim\operatorname{Bernoulli}(0.30),\qquad
X_{t,k}\sim\operatorname{Bernoulli}(0.35),\quad k=1,\ldots,6,
\]
so every true dose-control risk difference equals the clinical margin $0.05$.

The armwise union e-process analyzes each named dose-control comparison separately using the same one-sided half-normal mixture construction and a Bonferroni allocation $\alpha/K$; it declares global GO when any armwise e-process crosses $K/\alpha$. The fixed-final many-to-one max-$z$ test uses the maximum of the six standardized risk-difference statistics at $t=500$ with a shared-control simulated critical value $c_{\mathrm{final}}\approx2.31$. The planned five-look max-$z$ test uses looks $t=100,200,300,400,500$ with a simulated many-to-one repeated-look critical value $c_{5}\approx2.78$. The two misuse rows apply these planned thresholds at additional unplanned looks. The no-charge row uses the proposed anytime boundary but sets $\bar A_t^B(\delta_0)=0$.

Table~\ref{tab:type-1} shows the Type I error behavior for different methods. The first two rows are anytime-valid procedures and may stop at any synchronized update. The next two rows show planned Dunnett-style analyses used only at the looks for which their thresholds were calibrated. These rows are valid for their stated schedules, but they are not anytime-valid. The red rows then show what happens when the same planned thresholds are reused at additional, unplanned looks. Under the all-boundary null, this reuse changes the monitored statistic from a maximum over doses at planned times into a maximum over doses and many extra times, and the false-GO probability rises accordingly. The final row shows the necessity of the selection charge.

\begin{table}[H]
\centering
\caption{Simulation demonstration that planned Dunnett-style thresholds are not anytime-valid. All rows use the all-boundary null above, with $60{,}000$ independent trials through $500$ subjects per arm. Rows labeled ``misuse'' intentionally apply those planned thresholds at unplanned looks.}
\label{tab:type-1}
\scriptsize
\renewcommand{\arraystretch}{1.12}
\begin{tabularx}{\textwidth}{@{}>{\raggedright\arraybackslash}X>{\raggedright\arraybackslash}p{0.28\textwidth}c>{\raggedright\arraybackslash}p{0.16\textwidth}@{}}
\toprule
Analysis rule & When may it stop? & Anytime-valid? & False GO by 500\\
\midrule
Proposed selection-premium rule & Any synchronized update & Yes & $1.35\%$ ($809/60{,}000$)\\
Armwise union e-process & Any synchronized update & Yes & $0.80\%$ ($479/60{,}000$)\\
Fixed-final many-to-one max-$z$ test & Only final look $t=500$ & No & $4.78\%$ ($2{,}870/60{,}000$)\\
Planned five-look many-to-one max-$z$ test & Planned looks only & No & $4.53\%$ ($2{,}720/60{,}000$)\\
\textcolor{clinicalred}{Misuse: five-look boundary checked every 10 subjects from $t=100$} & Extra unplanned looks & No & \textcolor{clinicalred}{$7.77\%$ ($4{,}664/60{,}000$)}\\
\textcolor{clinicalred}{Misuse: five-look boundary checked after every update from $t=100$} & Continuous extra looks & No & \textcolor{clinicalred}{$9.39\%$ ($5{,}636/60{,}000$)}\\
\textcolor{clinicalred}{Misuse: fixed-final max-$z$ cutoff checked every 10 subjects from $t=100$} & Extra unplanned looks & No & \textcolor{clinicalred}{$20.91\%$ ($12{,}546/60{,}000$)}\\
\textcolor{clinicalred}{Misuse: fixed-final max-$z$ cutoff checked after every update from $t=100$} & Continuous extra looks & No & \textcolor{clinicalred}{$24.27\%$ ($14{,}562/60{,}000$)}\\
\textcolor{clinicalred}{Misuse: same anytime boundary but no selection charge} & Any update & No & \textcolor{clinicalred}{$7.26\%$ ($4{,}358/60{,}000$)}\\
\bottomrule
\end{tabularx}
\end{table}

Table~\ref{tab:traditional-power} reports performance for the valid uses of the same comparators. The armwise union e-process is a strong comparator: in these simulations it has lower all-boundary false-GO probability and higher GO-by-500 probability in the moderate and strong alternatives, although it is slightly lower in the near-threshold alternative. This comparison is important because it shows that the proposed method is not offered as a uniformly most powerful procedure. Its distinct output is a single selected-maximum ledger and an anytime lower confidence sequence for $\delta_*$.

\begin{table}[H]
\centering
\caption{Stopping behavior of valid uses of the methods. Each cell reports ``GO by 500; mean stop-or-cap enrollment per arm.'' Each row uses the same $20{,}000$ simulated trials per scenario as Table~\ref{tab:power-grid}.}
\label{tab:traditional-power}
\scriptsize
\renewcommand{\arraystretch}{1.12}
\begin{tabularx}{\textwidth}{@{}p{0.32\textwidth}XXX@{}}
\toprule
Analysis rule
& Near threshold
& Moderate
& Strong\\
& $\delta_*=0.10$
& $\delta_*=0.15$
& $\delta_*=0.20$\\
\midrule
Proposed selection-premium rule
& $11.5\%;\ 472.7$
& $57.1\%;\ 369.5$
& $94.8\%;\ 223.2$\\

Armwise union e-process
& $10.0\%;\ 478.9$
& $60.3\%;\ 369.9$
& $97.0\%;\ 217.7$\\

Fixed-final many-to-one max-$z$ test
& $38.5\%;\ 500.0$
& $91.0\%;\ 500.0$
& $99.9\%;\ 500.0$\\

Five-look many-to-one max-$z$ test
& $29.3\%;\ 444.2$
& $84.4\%;\ 308.0$
& $99.6\%;\ 190.3$\\
\bottomrule
\end{tabularx}
\end{table}

The previous tables describe finite-horizon operating characteristics through $500$ subjects per arm. Figure~\ref{fig:convergence-clear} instead checks the long-run mechanism behind Corollary~\ref{cor:consistency}. Under the strong profile, one dose has the unique largest response probability, so eventually the empirical leader separates from the other doses. Once this happens, the binary tie premium is usually zero, because $r_t=1$. The figure tracks whether the three asymptotic claims become visible numerically: the selection charge per subject should vanish, the anytime lower floor should approach the true best effect, and the normalized log e-value should approach its theoretical growth rate.

\begin{figure}[H]
\centering

\includegraphics[width=\textwidth]{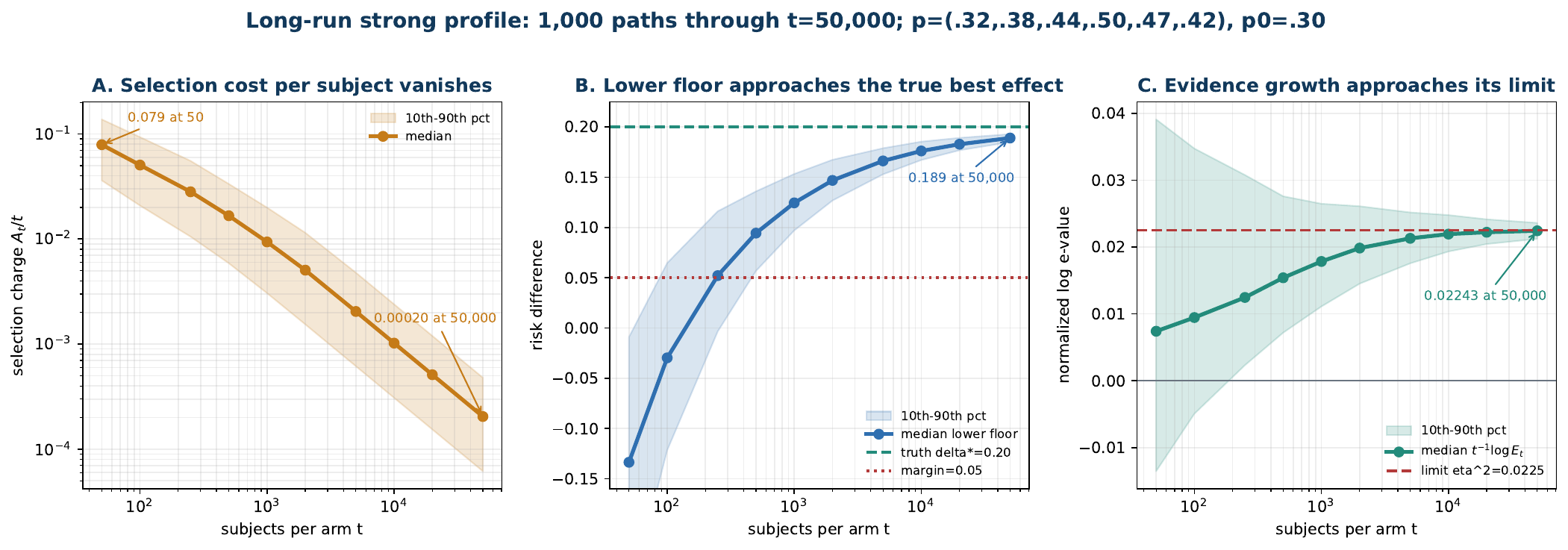}%

\caption{The exact strong-profile probabilities are $(0.32,0.38,0.44,0.50,0.47,0.42)$ with control probability $0.30$. Over $1{,}000$ paths through $50{,}000$ subjects per arm, Panel A shows the selection charge per subject becoming negligible, Panel B shows the median anytime lower confidence floor rising toward the true best effect $0.20$, and Panel C shows the median normalized log e-value approaching the theoretical evidence-growth limit $(0.20-0.05)^2=0.0225$. Shaded bands are the 10th-90th percentiles.}
\label{fig:convergence-clear}
\end{figure}

\section{Limitations and extensions}\label{sec:limits}

\paragraph{Dose-response trend information}
As noted in Introduction Section~\ref{sec:current-approaches}, our current approach does not account for dose-response trends. While this paper focuses on pairwise comparisons to clearly illustrate the methodology, incorporating trend information is an important direction for future work.

\paragraph{Balanced information blocks.}
The present theorems use synchronized blocks with one observation per active dose and control. Unequal randomization, delayed outcomes, covariate adjustment, arm addition, and arm dropping require an information-time or weighted-score extension.

\paragraph{Gaussian scale specification.}
The Gaussian theorem requires a prespecified covariance or a deterministic upper bound. Re-estimating variance from the same unblinded efficacy stream and inserting it without additional protection is not justified by the current proof. Independent or blinded variance adaptation is a natural extension.

\paragraph{Binary nuisance robustness.}
Maximizing over the unknown boundary rate is safe but can be conservative.

\paragraph{Global rather than named-dose inference.}
The lower confidence sequence concerns $\delta_*$. It does not establish that the empirical leader is the true best dose or quantify its selected-arm effect. A separate localization procedure remains necessary.

\paragraph{No universal dominance.}
The method can be more or less efficient than existing procedures depending on $K$, signal strength, score separation, endpoint distribution, and monitoring plan. The benefit of the proposed procedure is interpretable, state-adaptive selection accounting under flexible monitoring.

\section{Conclusion}

This paper shows how the selection-premium identity of \citet{de2026selection} can be turned into an anytime-valid procedure for dose-ranging trials. The key observation is that the maximum over dose-control contrasts has a predictable upward drift: before the next information block is observed, retaining the option to re-select the empirical leader has positive expected value, the selection premium. At an equal-margin reference it gives the exact expected winner optimism, and under the composite null it provides a conservative drift charge for the selected maximum.

Subtracting this predictable charge separates the two sources of distortion in a dose-ranging analysis. The cumulative premium $A_t/n_t$ accounts for selection optimism from reporting the current empirical winner, while the mixture boundary $q_\alpha(V_t;\rho)/n_t$ accounts for repeated monitoring over time. The resulting decision rule has the effect-scale form
\[
\widehat\delta_{\max,t}
\ge
\delta_0+\frac{A_t(\delta_0)}{n_t}
+\frac{q_\alpha(V_t;\rho)}{n_t}.
\]
Thus a global GO decision is made only when the raw best observed effect clears the clinical margin after paying both the selection charge and the anytime monitoring margin.

For Gaussian shared-control outcomes, the selection charge is computed from the observed cumulative dose sums and a prespecified covariance bound. For binary outcomes, the charge has a closed form depending on the number of current tied leaders; a unique leader contributes no new binary selection premium, while early ties produce a cumulative ledger that is later diluted on the effect scale. In both cases, the corrected residual yields an e-process and therefore controls the probability of ever making a false global GO decision. Inverting the candidate-margin tests gives an anytime lower confidence sequence for $\delta_*$, the best true effect among the tested doses. The simulations support the theoretical message.

The method developed in this paper gives a program-level proof-of-activity statement: it tests whether at least one studied dose exceeds the clinical margin and provides a lower confidence sequence for the best effect in the tested range. It does not identify the true best dose, give selected-arm inference for the empirical leader, or replace dose-response modeling, safety assessment, pharmacokinetic analysis, or benefit-risk evaluation. Its value is a mathematically explicit and operationally transparent way to make winner-bias-aware, anytime-valid global decisions in exploratory dose-ranging trials.

\bibliographystyle{plainnat}

\bibliography{refs}
\end{document}